\documentclass[11pt]{article}

\usepackage{arxivbasic}

\usepackage{graphicx}
\usepackage[title]{appendix}  

\usepackage{amsmath, amsthm, amssymb} 
\allowdisplaybreaks 

\newtheorem{theorem}{Theorem}

\newtheorem{corollary}{Corollary} 

\theoremstyle{definition}

\usepackage{dsfont} 

\usepackage{enumitem}

\usepackage[font=small,labelfont=bf]{caption}

\usepackage[mathscr]{eucal}

\usepackage{microtype}
 \DisableLigatures{encoding = *, family = * }

\usepackage{hyperref}
\hypersetup{colorlinks=true, citecolor=black, linkcolor=black, urlcolor=black}

\usepackage[
	backend=bibtex,
	style=authoryear,
	citestyle=authoryear-comp,
	natbib=true,
	sorting=nyvt,
	doi=true,isbn=false,url=false
]{biblatex}
\begin{document}

\title{Time Is Of The Essence:\\ Incorporating Phase-Type Distributed Delays And Dwell Times Into ODE Models}

\author{
    Hurtado, Paul J. \\
	University of Nevada, Reno \\
	ORCID: 0000-0002-8499-5986 \\
	\texttt{phurtado@unr.edu} \\
	\And
	Richards, Cameron\\
	University of Nevada, Reno\\ 
	ORCID: 0000-0002-1620-9998
}


\maketitle

\begin{abstract} Ordinary differential equations (ODE) models have a wide variety of applications in the fields of mathematics, statistics, and the sciences. Though they are widely used, these models are sometimes viewed as inflexible with respect to the incorporation of time delays. The Generalized Linear Chain Trick (GLCT) serves as a way for modelers to incorporate much more flexible delay or dwell time distribution assumptions than the usual exponential and Erlang distributions. In this paper we demonstrate how the GLCT can be used to generate new ODE models by generalizing or approximating existing models to yield much more general ODEs with phase-type distributed delays or dwell times. 	 	
\end{abstract}

\clearpage
\tableofcontents
\clearpage

\section{Introduction} \label{sec:intro} 

Ordinary differential equations (ODE) models are widely used in the sciences \citep[e.g., see][]{Strogatz2014, Murdoch2003, BeuterGlassMackeyTitcombe2003,Diekmann2000, AndersonMay1992}, but are often limited by the difficulty of incorporating time delays. Such delays are more easily incorporated into models using other mathematical frameworks. For example, delays can be modeled using integral equations or integro-differential equations to incorporate distributed delays into dynamic models, or using delay differential equations (DDEs) to incorporate fixed delays \citep{Wearing2005,Feng2016,Feng2007,Ruan2006,Lin2018,Roussel1996}. In an ODE framework, the linear chain trick has long been used to incorporate exponential and Erlang distributed delays \citep{Metz1991,Metz1986,Nisbet1989,Smith2010,Diekmann2017}, and recently this approach has been generalized to a much broader family of delay or dwell time distributions \citep{Hurtado2019}. 

This generalized linear chain trick \citep[GLCT;][]{Hurtado2019} allows modelers to incorporate delay and dwell time distributions that include  (but are not limited to) the phase-type family of univariate probability distributions \citep{Bladt2017ch3,Horvath2016,Reinecke2012a,BuTools2}, which include hyperexponential, hypoexponential, Coxian, and the previously mentioned exponential and Erlang distributions. The phase-type family of distributions is the set of all possible absorption time distributions for continuous time Markov chains (CTMCs) with one or more transient states and a single absorbing state. The GLCT also permits the use of similar time-varying versions of such distributions \citep{Hurtado2019}. Together, the GLCT and the tools and techniques associated with phase-type distributions enable modelers to draw from a richer set of ODE model assumptions when constructing new models, and provide a framework for more clearly seeing how underlying stochastic model assumptions are reflected in the structure of mean field ODE models.

In this paper, we illustrate how to use the GLCT to formulate new ODE models that explicitly incorporate phase-type distributed delays. We do this by generalizing multiple different models from the literature. These include ODE models with no explicit delays, DDE models, and distributed delay models in the form of integro-differential equations. In section \ref{sec:glct} we review the GLCT framework for phase-type distributions, as well as the standard Linear Chain Trick (LCT). We then generalize multiple models starting in section \ref{sec:TGI}, where we generalize a tumor growth inhibition model by \citet{Simeoni2004}. In section \ref{sec:opioid} we then generalize a prescription opioid epidemic model by \citet{Battista2019}, then a within-host immune-pathogen model by \citet{Hurtado2012} in section \ref{sec:inhost}, and finally a model of cell-to-cell spread of HIV by \citet{Culshaw2003} in section \ref{sec:HIV}.

\subsection{Generalized Linear Chain Trick}\label{sec:glct}
For our purposes below, we here provide a statement of the Generalized Linear Chain Trick (GLCT) for phase-type distributions. More generally, the GLCT in \citet{Hurtado2019} extends the version below to also include time-varying parameters (analogous to extending homogeneous Poisson process rates to the time-varying rates of inhomogeneous Poisson processes). 

The (continuous) phase-type distributions are a family of matrix exponential distributions that can otherwise be thought of as the absorption time distributions for CTMCs with a single absorbing state. They are parameterized in terms of a vector \textbf{v}, which is the initial distribution vector across the set of transient states\footnote{The full initial distribution vector would be [\textbf{v}, $v_0$].}, and the transient state block \textbf{M} of the transition rate matrix, which together define the corresponding CTMC. The general form for the density function, cumulative distribution function, and moments of a phase-type distribution are  \begin{subequations} \begin{align}
	f(t) =&\; \text{\textbf{v}}\,e^{\mathbf{M}t}\,(-\mathbf{M}\mathbf{1}) \\
	F(t) =&\; 1 - \text{\textbf{v}}\,e^{\mathbf{M}t}\,\mathbf{1} \\
	E(T^j)=&\; j!\,\text{\textbf{v}}\,(-\mathbf{M})^{-j}\mathbf{1} 
\end{align} \end{subequations} where $\mathbf{1}$ is an appropriately long column vector of ones.

For more on phase-type distributions, see \citet{HurtadoRichards2020, Hurtado2019, Bladt2017ch3, BuTools2, Horvath2016, Reinecke2012a}. 

\begin{theorem}[\textbf{GLCT for phase-type distributions}] \label{th:glct}
	Assume individuals enter a state (call it state X) at rate $\mathcal{I}(t)\in\mathbb{R}$ and that the distribution of time spent in state X follows a continuous phase-type distribution given by the length $k$ initial probability vector $\textnormal{\textbf{v}}$ and the $k\times k$ matrix $\mathbf{M}$. Then partitioning X into $k$ sub-states X$_i$, and denoting the corresponding amount of individuals in state X$_i$ at time $t$ by $x_i(t)$, then the mean field equations for these sub-states $x_i$ are given by
	\begin{equation} \frac{d}{dt}\mathbf{x}(t)=\textnormal{\textbf{v}}\,\mathcal{I}(t) + \mathbf{M}^\text{T}\,\mathbf{x}(t) \label{eq:GLCT}\end{equation}
	where the rate of individuals leaving each of these sub-states of X is given by the vector $(-\mathbf{M\,1})\circ\mathbf{x}$, where $\circ$ is the Hadamard (element-wise) product of the two vectors, and thus the total rate of individuals leaving state X is given by the sum of those terms, i.e., $(-\mathbf{M\,1})^\text{T}\mathbf{x}=-\mathbf{1}^\text{T}\mathbf{M}^\text{T}\mathbf{x}$.
\end{theorem}


The standard Linear Chain Trick (LCT) is well known (see \citet{Hurtado2019} and references therein) and is a special case of Theorem \ref{th:glct} above. However, it is usually stated without the above matrix-vector notation. The following is a formal statement of the standard LCT, but here we have slightly generalized it to include generalized Erlang distributions (i.e., the sum of $k$ independent exponentially distributed random variables, each with potentially different rates $r_i$) as this only changes a few subscripts in the mean field equations. See \citet{Smith2010} for a similar statement of the standard LCT (for Erlang distributions).

\begin{corollary}[\textbf{Linear Chain Trick for Erlang and Hypoexponential Distributions}] \label{th:lct}
	~\\
	Consider the GLCT above. Assume that the dwell-time distribution is a generalized Erlang (hypoexponential) distribution with rates $\mathbf{r}=[r_1,\;r_2,\;\ldots\;,r_k]^\text{T}$, where $r_i>0$, or an Erlang distribution with rate $r$ (all rates $r_i=r$) and shape $k$ (or if written in terms of shape $k$ and mean $\tau=k/r$, use $r=k/\tau$). Then the corresponding mean field equations are 
		\begin{equation} \label{eq:LCT}\begin{split} 
		\frac{dx_1}{dt} =&\; \mathcal{I}(t) - r_1\,x_1 \\
		\frac{dx_2}{dt} =&\;  r_1\,x_1 - r_2\,x_2 \\
		&\vdots \\
		\frac{dx_k}{dt} =&\; r_{k-1}\,x_{k-1} - r_{k}\,x_{k}.
		\end{split}\end{equation}	
\end{corollary}
\begin{proof}
	The phase-type distribution formulation of the generalized Erlang distribution described above is given by \begin{equation} \label{eq:vM} \textnormal{\textbf{v}} = \begin{bmatrix} 1 \\ 0 \\ \vdots \\ 0  \end{bmatrix}  \qquad \text{and} \qquad \mathbf{M} = \begin{bmatrix} -r_1 & r_1 & 0 & \cdots & 0 \\
	0 & -r_2 & r_2  & \ddots & 0 \\
	\vdots & \ddots & \ddots  & \ddots & \ddots \\
	0 & 0 & \ddots &  -r_{k-1} & r_{k-1} \\
	0 & 0 & \cdots &  0 & -r_k
	\end{bmatrix}. \end{equation} Substituting these into eq. \eqref{eq:GLCT} yields the desired result. If all $r_i=r$ then the phase-type distribution is an Erlang distribution with rate $r$ and shape $k$ (mean $k/r$ and coefficient of variation $1/\sqrt{k}$).
\end{proof}

\subsubsection{Using The GLCT To Derive New ODEs With Phase-Type Distributed Delays}

As we will show below, it is relatively straightforward to derive, from an existing ODE model with exponential or Erlang delays, a more general ODE model with phase-type distributed delays. Such models can also be similarly derived from existing integral equations \citep{Hurtado2019}. In this section, we describe how modelers can also use the GLCT to derive ODEs with phase-type distributed delays from existing ODEs and DDEs. 

One important application of the LCT is that it can be used to approximate delay differential equations (DDEs) with ODEs, as discussed in \citet{Smith2010}. To do this, DDEs can be thought of in the context of distributed delay models as the result of assuming a delay distribution with a point mass at $\tau$. This distribution can be approximated by an Erlang distribution with mean $\tau$ and a very small coefficient of variation (i.e., a large shape parameter), which yields an ODE via the LCT. That Erlang distribution assumption can then be replaced with a more general phase-type distribution assumption to yield ODEs via the GLCT, as illustrated in the following example. 

\textbf{Example:} Consider a simple birth-death model where the recruitment rate ($f$) of new adults at time $t$ is a function of the number of adults $\tau$ time units in the past. Suppose $x$ tracks the number of adults in a population, that there is an assumed maturation time of $\tau$ time units, and that adults die with per-capita rate $g(x)$. Then the dynamics of $x$ could be modeled by the DDE 

\begin{equation} 
\frac{dx}{dt} =\; f(x(t-\tau)) - g(x(t))\,x(t).
\end{equation}

The LCT can be used to approximate the above DDE with an ODE, by replacing the assumption of a fixed time delay (or maturation time, in this example) with an Erlang distributed delay, which can have arbitrarily small variance. This is accomplished by using a shape parameter $k\gg1$ (recall the coefficient of variation is given by $1/\sqrt{k}$), and assuming a rate $r=k/\tau$ which yields the desired Erlang distribution with mean $r/k=\tau$. By the LCT, this introduces $k$ new state variables to the model, which track this delayed quantity, and in this context these new states can be thought of as a sequence of immature stages. Applying the LCT in this manner yields the following ODE approximation of the above DDE, where these immature stages are denoted $w_i$, $i=1,\ldots,k$. \begin{subequations} \label{eq:DDE2ODE}
\begin{align}
\frac{dw_1}{dt} =&\; f(x) - r\,w_1 \\
\frac{dw_j}{dt} =&\; r\,w_{j-1}  - r\,w_j, \quad\; j=2,\ldots,k \\
\frac{dx}{dt} =&\;  r\,w_k - g(x)\,x 
\end{align}\end{subequations}

This approximation can then be used as an intermediate step to derive a phase-type distributed delay model. To do this, simply write the above system of ODEs in matrix form \textit{a la} the GLCT,  \begin{subequations} \label{eq:DDE2PTODE}
\begin{align}
\frac{d\mathbf{w}}{dt} =&\; f(x)\text{\textbf{v}} + \mathbf{M}^\text{T} \mathbf{w} \\
\frac{dx}{dt} =&\;  -\mathbf{1}^\text{T}\mathbf{M}^\text{T}\mathbf{w} - g(x)\,x,
\end{align}\end{subequations}

where the term $r\,w_k$ is instead written using the more generic expression given in Theorem \ref{th:glct} for computing the overall loss rate from the intermediate states $\mathbf{w}$, which in this case is $-\mathbf{1}^\text{T}\mathbf{M}^\text{T}\mathbf{w} = r\,w_k$, and the vector \textbf{v} and matrix \textbf{M} are of the form of eqs. \eqref{eq:vM}, with $r_i=k/\tau$.

The above steps leading to the derivation of eqs. \eqref{eq:DDE2PTODE} show how a DDE can be approximated by changing the fixed delay assumption to an Erlang distributed delay, and then generalized to a phase-type distributed delay. In this more general form, \textbf{v} and \textbf{M} can correspond to any phase-type distribution, not just an Erlang distribution. Furthermore, this assumption could be relaxed even further by allowing entries in \textbf{v} or \textbf{M} to vary over time, or with one or more state variables, as described in \citet{Hurtado2019}.

\section{Results}\label{sec:results}

In the sections below, we illustrate this process of deriving new models using the GLCT by generalizing various biological models taken from the peer reviewed literature. We first highlight some of the key assumptions of these models related to delays and the time spent in different states, viewing each model as a mean field model corresponding to some unspecified stochastic model. We then derive from each model an ODE model that incorporates phase-type distributed delays or dwell times, thereby generalizing or approximating the original model.

\subsection{Model 1: Tumor Growth Inhibition (TGI) Model} \label{sec:TGI} 

\begin{figure}[!tbh]
	\centering\includegraphics[trim = 0 450 430 0, clip, width=0.8\textwidth]{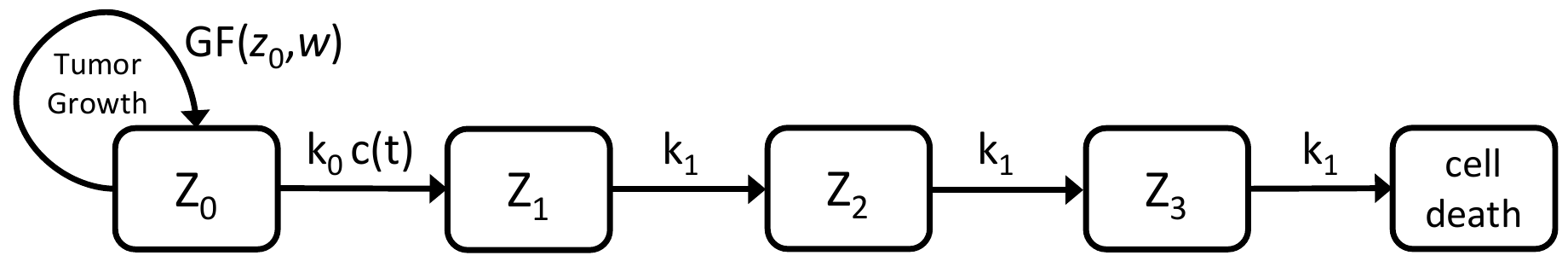}
	\caption{Schematic diagram of the Tumor Growth Inhibition model (TGI) by \citet{Simeoni2004}. See the main text for further details. \label{fig:TGI}}
\end{figure}

\citet{Simeoni2004} introduced a simple model of tumor growth inhibition that employs the standard Linear Chain Trick (LCT) to incorporate an Erlang distributed time to cell death following tumor cell damage from treatment. The model was subsequently analyzed using standard approaches from dynamical systems \citep{Magni2006,Magni2008}, and has been used elsewhere in the study of tumor growth and the development of cancer treatments \citep[e.g.,][]{Rocchetti2009,Simeoni2013}.  The Simeoni model has also been extended to a ``competing Poisson processes" like assumption (compare Fig. 6 in \citet{Hurtado2019} to Fig. 2 in \citet{Rocchetti2009} and Fig. 1 in \citet{Terranova2013}) in order to model tumor cell death arising from the combined effects of two drugs with no pharmacokinetic interaction.

The basic TGI model is given by eqs. \eqref{eq:GF} and \eqref{eq:TGI} below. In the absence of pharmacological treatment, the amount of cycling (replicating) tumor cells at time $t$ ($z_0(t)$) grows according to the overall growth rate\footnote{This growth rate function is an approximation of the piece-wise function that is equal to $\lambda_0 z_0$ when $w<\lambda_0/\lambda_1$, and $\lambda_1 z_0/w$ when $w\geq \lambda_0/\lambda_1$. See \citet{Magni2006} for details.} \begin{equation} \label{eq:GF}
GF(z_0,w) =\; \frac{\lambda_0\,z_0(t)}{\bigg[1 + \bigg(\frac{\lambda_0}{\lambda_1}\,w(t)\bigg)^\psi \bigg]^\frac{1}{\psi}}.
\end{equation}

Treatment begins at time $t_0>0$, and accordingly the effect of that treatment $c(t)=0$ for $0\leq t\leq t_0$. Once treatment begins, cells that are damaged by the treatment then progress through a series of states Z$_i$, $i=1,\ldots,n$, prior to cell death (see Fig. \ref{fig:TGI}). Together, the full model is given by \begin{subequations}\label{eq:TGI} \begin{align}
\frac{dz_0(t)}{dt} =&\; GF(z_0(t),w(t)) - k_0\,c(t)\,z_0(t) \\
\frac{dz_1(t)}{dt} =&\;  k_0\,c(t)\,z_0(t) - k_1\,z_1(t) \\
\frac{dz_i(t)}{dt} =&\;  k_1\,z_{i-1}(t) - k_1\,z_i(t), \quad\; i = 2,\ldots,n \\
w(t) =&\; \sum_{i=0}^n z_i; \;\; z_0(0)=w_0,\;\; z_i(0)=0, \; i=1,\ldots,n
\end{align}
\end{subequations}

where $w$ is the total amount of tumor cells, and $k_0$ and $c(t)\geq0$ determine the rate of initial tumor cell damage from the treatment. Alternatively, from a more mathematical perspective, $k_0$ and $c(t)$ determine the distribution of time spent in the base state Z$_0$, which follows the first event time distribution under a non-homogeneous Poisson process with rate $r(t)=k_0c(t)$ (see \citet{Hurtado2019} for details). Parameters $n$ and $k_1$ are the shape and rate parameters, respectively, for the Erlang distributed time until cell death for the cells damaged by the treatment. The treatment is assumed to have no effect on the time until cell death after the initial damage to the cell.

To extend this model to instead assume a more general phase-type distributed time to cell death, the equations for $z_i$, $i=1,\ldots,n$ in eqs. \eqref{eq:TGI} can be written in matrix form, using Theorem \ref{th:glct}, where \begin{equation}\textnormal{\textbf{v}} = \begin{bmatrix} 1 \\ 0 \\ \vdots \\ 0  \end{bmatrix}  \qquad \text{and} \qquad \mathbf{M} = \begin{bmatrix} -k_1 & k_1 & 0 & \cdots & 0 \\
0 & -k_1 & k_1  & \ddots & 0 \\
\vdots & \ddots & \ddots  & \ddots & \ddots \\
0 & 0 & \ddots &  -k_{1} & k_{1} \\
0 & 0 & \cdots &  0 & -k_1
\end{bmatrix}. \end{equation}

This yields the more compact, and more general, set of equations below, where $\mathbf{x}=[z_1,z_2,\ldots,z_n]^\text{T}$. \begin{equation}\label{eq:TGIglct} \begin{split}
\frac{dz_0(t)}{dt} =&\; GF(z_0(t),w(t)) - k_0\,c(t)\,z_0(t) \\
\frac{d\mathbf{x}(t)}{dt} =&\;  k_0\,c(t)\,z_0(t)\,\text{\textbf{v}} + \mathbf{M}^\text{T}\mathbf{x} \\
w(t) =&\; \sum_{i=0}^n z_i \\
z_0(0)=&w_0,\; z_i(0)=0, \; i=1,\ldots,n.
\end{split}
\end{equation}

Note that eqs. \eqref{eq:TGIglct} generalize the TGI model in the sense that these equations accommodate any phase-type distribution assumption for the time to cell death following the initial effect of treatment, not just the Erlang distribution assumed in the original TGI model. Additionally, this matrix-vector form of the original TGI model (i.e., assuming an Erlang distribution) can still be used with some benefit for both computational and mathematical analyses of the TGI model given by eqs. \eqref{eq:TGI}, where those analyses can take advantage of the matrix-vector form of these more general equations \citep{HurtadoRichards2020}. 

\subsection{Model 2: Perscription Opioid Epidemic Model}\label{sec:opioid}

The prescription opioid epidemic model by \citet{Battista2019} is a system of ordinary differential equations with no explicit time delays, and (implicit) exponentially distributed dwell times in multiple states. The model assumes individuals are in one of four different states: S, P, A, and R. Here $S$ is the size of the susceptible class. These individuals are not using opioids or recovering from addiction. $P$ is the number of prescribed users (those who are prescribed the drugs and using them but have no addiction). $A$ is the number of addicted individuals who can be using either prescribed or ilicit opioids, and $R$ is the number of individuals undergoing treatment to recover from addiction. 

The model describing how individuals transition among these states is given by the following equations, where the dot over each state variable indicates a time derivative. \begin{subequations}\label{eq:opiod}\begin{align}
\dot{S}&= - \alpha S - \beta _A SA - \beta _P SP + \epsilon P + \delta R + \mu (P+R)+\mu^{*}A \\
\dot{P}&= \alpha S - (\epsilon +\gamma +\mu ) P \label{eq:opiodP}\\
\dot{A}&= \gamma P + \sigma R + \beta _A SA + \beta _P SP - (\zeta +\mu ^{*})A \\
\dot{R}&= \zeta A - (\delta +\sigma +\mu ) R. 
\end{align}\end{subequations}

The term $\alpha S$ is the rate of individuals transitioning from the susceptible state to the prescribed state after being prescribed opioids per unit time, $\beta _A S A$ is the rate of those transitioning from state S to state A after interacting with addicted individuals, and similarly $\beta _P SP$ represents the rate of individuals who transition from the susceptible class to the addicted class after exposure to opioids via perscription opiod users who have extra or unsecured drugs. The terms $\epsilon P$ and $\delta R$ are the rate individuals leave the prescribed users class without becoming addicted and then reenter the susceptible class at per-capita rate $\epsilon$, and those who leave the rehabilitation state after treatment and reenter the susceptible state at per-capita rate $\delta$. The rates $\mu P$, $\mu R$ and $\mu^*A$ are the death rates for the prescribed, rehabilitated, and addicted classes (to ensure constant population size, deaths are replaced instantaneously by new susceptible individuals). The term  $\gamma P$ is the rate that individuals leave the prescribed class by becoming addicted to their prescription opioids, $\zeta A$ is the rate at which addicted individuals initiate treatment, and $\sigma R$ is the rate at which individuals who are undergoing treatment reenter the addiction class.

Note that from the model terms described above, we may assume that prescription users remain in the prescribed state for an exponentially distributed amount time \citep{Hurtado2019}. Thus, another way to interpret these model terms is that (focusing on eq. \eqref{eq:opiodP}, for example) the proportion of individuals which leave the prescribed user state and go to the susceptible state is $\frac{\epsilon}{\epsilon+\gamma+\mu}$, and thus the net rate of individual entering the susceptible state from the prescribed state is $\frac{\epsilon}{\epsilon+\gamma+\mu}\,(\epsilon+\gamma+\mu)P=\epsilon\,P$. Similarly, the proportion of individuals who go on to become addicted, and who die, are given by $\frac {\gamma}{\epsilon+\gamma+\mu}$ and $\frac {\mu}{\epsilon+\gamma+\mu}$, respectively.

To generalize this model, we aim to replace the implicit assumption of exponentially distributed dwell times in each state, and replace those with the more general phase-type distributions instead.


If we assume the dwell time distribution for the prescribed user state P is a continuous phase-type distribution parameterized by the $n\times 1$ parameter vector $\bf{v_P}$ and $n\times n$  matrix $\bf{M_P}$, then to total number of individuals in state P is given by the sum of the $n$ sub-states $P_i$, $i=1,\ldots, n$. Let ${\bf{x}}=\left[P_1,P_2, \dots, P_n\right]^{\text{T}}$. Then by the GLCT (Theorem \ref{th:glct}), the mean field equations for our prescribed user sub-states are  \begin{equation} \label{eq:dx} \dot{{\bf{x}}} = {\bf{v_P}}\alpha S + {\bf{M_P}}^{\text{T}} {\bf{x}}.\end{equation}

Observe that if we let $\bf{v_P}$ be a one dimensional row vector with its first and only entry being a 1 and let $\bf{M_P}$ be a $1\times 1$ matrix with the entry $-\left(\epsilon +\gamma +\mu \right)$, we recover eq. \eqref{eq:opiodP}. 

Recall that individuals who leave the prescribed user state either transition to the addicted state, the susceptible state, or they die. We can denote these proportions as $F_{PA}$, $F_{PS},$ and $F_{PD}$, respectively, where $F_{ij}\in[0,1]$ and $F_{PA}+F_{PS}+F_{PD}=1$. Note that in the original model $F_{PA}=\frac{\gamma}{(\epsilon+ \gamma +\mu)}, F_{PS}=\frac{\epsilon}{(\epsilon+ \gamma +\mu)},$ and $F_{PD}=\frac{\mu}{(\epsilon+ \gamma +\mu)}$. Combining the above with eq. \eqref{eq:dx}, this yields \begin{subequations}\label{eq:opiod2}\begin{align}
\dot{S}=&\; - \alpha S - \beta _A SA - \beta _P SP + F_{PS} (-{\bf{M_P 1}})^{\text{T}} {\bf{x}} + \delta R +F_{PD} (-{\bf{M_P 1}})^{\text{T}} {\bf{x}}+\mu R +\mu ^{*}A \\
\dot{{\bf{x}}}=&\; {\bf{v_P}}\alpha S + {\bf{M_P}}^{\text{T}} {\bf{x}} \\
\dot{A}=&\; F_{PA} (-{\bf{M_P 1}})^{\text{T}} {\bf{x}} + \sigma R + \beta _A SA + \beta _P SP - (\zeta +\mu ^{*})A \\
\dot{R}=&\; \zeta A - (\delta +\sigma +\mu ) R.
\end{align}\end{subequations}

Similarly, we can generalize the addicted and rehabilitated states with phase-type dwell time distributions, assuming the respective phase-type distributions are parameterized by $\mathbf{v_A}$, $\mathbf{M_A}$, $\mathbf{v_R}$, and $\mathbf{M_R}$. Let ${\bf{y}}=\left[A_1,A_2, \dots, A_k\right]^{\text{T}}$ denote the $k$ sub-states of $A$, and ${\bf{z}}=\left[R_1,R_2, \dots, R_m\right]^{\text{T}}$ the $m$ sub-states of $R$. This yields the generalized model:  \begin{subequations} \label{eq:opiodglct} \begin{align} \begin{split}
\dot{S}=&\; - \alpha S - \beta _A SA - \beta _P SP + (F_{PS}+F_{PD}) (-{\bf{M_P 1}})^{\text{T}} {\bf{x}} \;+ \\ 
        &\; + (F_{RS}+F_{RD}) (-{\bf{M_R 1}})^{\text{T}}\,{\bf{z}}  + F_{AD} (-{\bf{M_A 1}})^{\text{T}}\,{\bf{y}} \end{split} \\
\dot{{\bf{x}}} =&\; {\bf{v_P}}\alpha S + {\bf{M_P}}^{\text{T}} {\bf{x}} \\
\dot{{\bf{y}}}=&\; {\bf{v_A}}\big(F_{PA} (-{\bf{M_P 1}})^{\text{T}} {\bf{x}} +  F_{RA} (-{\bf{M_R 1}})^{\text{T}}{\bf{z}} + \beta _A SA + \beta _P SP\big) + {\bf{M_A}}^{\text{T}}\,{\bf{y}}\\
\dot{{\bf{z}}}=&\; {\bf{v_R}}(F_{AR} (-{\bf{M_A 1}})^{\text{T}}\,{\bf{y}})+ {\bf{M_R}}^{\text{T}} {\bf{z}}.
\end{align} \end{subequations}

It is worth noting that the original model eqs. \eqref{eq:opiod} are a special case of eqs. \eqref{eq:opiodglct}, as are any intermediate extensions of the original model obtained by applying the standard Linear Chain Trick (LCT) to impose Erlang distributed dwell times on one or more of the four main states.

\subsection{Model 3: Within-Host Model of Immune-Pathogen Interactions} \label{sec:inhost}

In \citet{Hurtado2012}, a specific (adaptive) immune response was added to the innate immune response model introduced by \citet{Reynolds2006}. The scaled version of this within-host model, as stated in \citet{Hurtado2012}, is  \begin{subequations}\label{eq:inhost}\begin{align}
	\frac{dp}{dt} =&\; k_{pg}p(1-p) - \frac{k_m p}{\mu_p+p} -K(y)\,n\,p \label{eq:inhostp} \\
	\frac{dn}{dt} =&\; \frac{n+k_p \,p}{x_n + n+k_p \,p} - \mu_n \, p \\
	\frac{dy_0}{dt} =&\; \frac{(np)^\alpha}{x_y^\alpha + (np)^\alpha} - \mu_{y0}\,y_0 \\
	\frac{dy}{dt} =&\; \mu_{y0}\,y_0  - \mu_y\,y
	\end{align}
\end{subequations}

In this model, $p$ is the scaled pathogen (bacteria) population size, which follows a logistic growth model in the absence of an immune response. The second term in eq. \eqref{eq:inhostp} models the effect of some baseline local immune defenses capable of neutralizing a small population of pathogen, and mathematically introduces a strong Allee effect into the model. The level of innate immune activity $n$ increases in response to the presence of pathogen, as well as from a positive feedback loop, and the interaction of this innate immune activity and pathogen stimulates progenitor cells ($y_0$) that mature into active specific immune components ($y$), e.g., B-cells, which augment the pathogen-killing capacity of the innate immune components (i.e., which increase $K(y)$).  For further details on this model, see \citet{Hurtado2012} and \citet{Reynolds2006}.

In this model, the delay in activating the specific immune response can be thought of as an exponentially distributed maturation time (with mean $1/\mu_{y0}$) and the duration of the active immune response (i.e., the dwell time of mature specific immune components modeled by $y$) is also exponentially distributed (with mean $1/\mu_y$).

Both of these dwell time distribution assumptions can be replaced by phase-type distributions with respective parameters $\mathbf{v_{y0}}$, $\mathbf{M_{y0}}$, $\mathbf{v_{y}}$ and $\mathbf{M_y}$, respectively.  To do this, we first partition state Y$_0$ into sub-states X$_i$, $i=1,\ldots,m$, and the state Y into sub-states Z$_j$, $j=1,\ldots,n$, where $y_0=\sum_{i=1}^m x_i$ and $y=\sum_{j=1}^n z_j$, and we let  $\mathbf{x}=[x_1,\ldots,x_m]^\text{T}$ and $\mathbf{z}=[z_1,\ldots,z_n]^\text{T}$. The GLCT (Theorem \ref{th:glct}) then yields the more general model \begin{subequations}\label{eq:inhostglct}\begin{align}
	\frac{dp}{dt} =&\; k_{pg}p(1-p) - \frac{k_m p}{\mu_p+p} -K(y)\,n\,p  \\
	\frac{dn}{dt} =&\; \frac{n+k_p \,p}{x_n + n+k_p \,p} - \mu_n \, p \\
	\frac{d\mathbf{x}}{dt} =&\; \mathbf{v_{y0}}\frac{(np)^\alpha}{x_y^\alpha + (np)^\alpha} + \mathbf{M_{y0}}^\text{T}\,\mathbf{x} \\
	\frac{d\mathbf{z}}{dt} =&\; -\mathbf{1}^\text{T}\mathbf{M_y}^\text{T}\mathbf{v_y} + \mathbf{M_y}^\text{T}\mathbf{z}.
	\end{align}
\end{subequations}

Note that the above model could be similarly be extended to include a phase-type distributed time lag in the activation of the non-specific immune response (modeled by $n$). However, given the relatively fast activation time of this response, we have omitted this extension in the above model.

\subsection{Model 4: Cell-To-Cell Spread of HIV}\label{sec:HIV}

\citet{Culshaw2003} introduced an integro-differential model for the cell-to-cell spread of HIV, which incorporates a distributed time delay in the time between cells becoming infected and infectious. They then derive from this general model multiple other models which differ only in the specific assumptions on the form of this delay distribution.

In their most general model, state variable $C(t)$ represents the concentration of healthy cells at time $t$, and $I(t)$ is the concentration of infected cells. The model is as follows: \begin{subequations} \label{eq:CIide}\begin{align}
	\frac {dC}{dt} =&\;  r_C C(t) \left(1-\frac{C(t)+I(t)}{C_M}\right) -k_IC(t)I(t) \\
	\frac {dI}{dt} =&\;  k'_i\int_{-\infty}^{t} C(u)I(u)F(t-u)du-\mu_{I}I(t). 
	\end{align} \end{subequations}

Parameter $r_C$ is the net growth rate of the healthy cell population, $C_M$ is an effective carrying capacity of the system, $k_I$ is an infection rate parameter, $k_I'/k_I$ is the fraction of cells surviving the incubation period, and $\mu_I$ is the per capita death rate of infected cells (implicitly, the infected cell lifetime is exponentially distributed with mean $1/\mu_I$). Initial values for $C$ and $I$ must be functions defined over all $s\in (-\infty,0]$ and are denoted $\phi(s)\geq 0$ and $\psi(s)\geq0$, respectively.

In \citet{Culshaw2003}, the delay kernel $F(u)$ is assumed to be of the form \begin{equation} \label{eq:CIF} F(u) =\; \frac{\alpha^{n+1}u^n}{n!}e^{-\alpha\,u} \end{equation} which is just the density function for an Erlang distribution with rate $\alpha$ and shape $n+1$ (and thus, mean $(n+1)/\alpha$ and coefficient of variation $1/\sqrt{n+1}$), and the \textit{weak} and \textit{strong} kernels are just the particular cases where the shape parameter is 1 (i.e., an exponential distribution with rate $\alpha$) or 2 (Erlang with rate $\alpha$ and shape 2), respectively.

Three models are then derived in \citet{Culshaw2003} from this more general integro-differential equation model, which we will summarize here then extend further using the LCT and GLCT (see \citet{Culshaw2003} for a comparison of the dynamics of these three models).

First, assuming $F(u)=\delta(u)$ is a Dirac delta function (point mass at zero) yields the model with no delay, given in \citet{Culshaw2003} by equations  \begin{equation}\begin{split}
\frac {dC}{dt} =&\;  r_C C(t) \left(1-\frac{C(t)+I(t)}{C_M}\right) -k_IC(t)I(t) \\
\frac {dI}{dt} =&\;  k'_i C(t)I(t) - \mu_{I}I(t),
\end{split} \end{equation}

with initial conditions $C(0)=c_0\geq 0$ and $I(0)=I_0\geq 0$.

Second, assuming $F(u)=\delta(u-\tau)$ is a Dirac delta function at time $\tau>0$ yields the delay differential equation below (as written in \citet{Culshaw2003}) with the same initial conditions as eqs. \eqref{eq:CIide}. \begin{equation}\begin{split}
\frac {dC}{dt} =&\;  r_C C(t) \left(1-\frac{C(t)+I(t)}{C_M}\right) -k_IC(t)I(t) \\
\frac {dI}{dt} =&\;  k'_i C(t-\tau)I(t-\tau) - \mu_{I}I(t). 
\end{split} \end{equation}

Third, assuming a ``weak kernel" (i.e., exponentially distributed delay with rate $\alpha$) yields \begin{equation}\label{eq:CXI}\begin{split}
\frac{dC}{dt} =&\;  r_C C(t) \left(1-\frac{C(t)+I(t)}{C_M}\right) -k_IC(t)I(t) \\
\frac{dX}{dt} =&\;  \alpha\,C(t)I(t) - \alpha\,X(t) \\
\frac{dI}{dt} =&\;  k'_I\,X(t) - \mu_{I}I(t). 
\end{split} \end{equation}

Observe that, by making the simple substitution $Y(t)=\frac{k_I}{\alpha}X(t)$, we can write the following alternative model which is equivalent to eqs. \eqref{eq:CXI} by Culshaw \textit{et al}, but is more natural in terms of the units of $X$ and $Y$ and in the context of the LCT: \begin{equation} \label{eq:CYI} \begin{split}
\frac{dC}{dt} =&\;  r_C C(t) \left(1-\frac{C(t)+I(t)}{C_M}\right) -k_IC(t)I(t) \\
\frac{dY}{dt} =&\;  k_I\,C(t)I(t) - \alpha\,Y(t) \\
\frac{dI}{dt} =&\;   \frac{k'_I}{k_I} \alpha Y(t) - \mu_{I}I(t). 
\end{split} \end{equation}

From eqs. \eqref{eq:CYI}, it is straightforward to derive two additional models using the LCT and GLCT. Using the LCT, eqs. \eqref{eq:CIlct} below are a more general form of eqs. \eqref{eq:CYI} that correspond to any choice of a non-negative integer value of $n$ and $\alpha>0$ for the delay kernel $F$ in eq. \eqref{eq:CIF} (i.e., any Erlang distribution with shape $n+1$ and rate $\alpha$).  \begin{equation} \label{eq:CIlct} \begin{split}
\frac{dC}{dt} =&\;  r_C C(t) \left(1-\frac{C(t)+I(t)}{C_M}\right) -k_IC(t)I(t) \\
\frac{dY_1}{dt} =&\;  k_I\,C(t)I(t) - \alpha\,Y_1(t) \\
\frac{dY_i}{dt} =&\;  \alpha\,Y_{i-1}(t) - \alpha\,Y_i(t), \quad\; i=2,\ldots,n+1 \\
&\vdots \\
\frac{dI}{dt} =&\;   \frac{k'_I}{k_I} \alpha Y_{n+1}(t) - \mu_{I}I(t). 
\end{split} \end{equation}

Using the above equations as a guide, re-writing them in matrix form as suggested by the GLCT (Theorem \ref{th:glct}), yields the more general set of equations below, which are the desired set of model equations for which the Erlang distribution assumption (with parameters $n+1$ and $\alpha$) has been replaced by a phase-type distribution parameterized by the length $k$ vector \textbf{v} and $k\times k$ matrix \textbf{M}, where $\mathbf{y}=[Y_1,\ldots,Y_k]^\text{T}$. \begin{equation} \label{eq:CIglct} \begin{split}
\frac{dC}{dt} =&\;  r_C C(t) \left(1-\frac{C(t)+I(t)}{C_M}\right) -k_IC(t)I(t) \\
\frac{d\mathbf{y}}{dt} =&\;  k_I\,C(t)I(t)\,\mathbf{v} + \mathbf{M}^\text{T}\,\mathbf{y}\\
\frac{dI}{dt} =&\;  -\frac{k'_I}{k_I}\mathbf{1}^\text{T}\,\mathbf{M}^\text{T}\,\mathbf{y} - \mu_{I}I(t). 
\end{split} \end{equation}

Further generalizations, e.g., to time-varying \textbf{v}, \textbf{M}, or survival fraction $f=k'_I/k_I$, are also possible \citep{Hurtado2019}.

\section{Discussion}\label{sec:discussion}

Mean field state transition models, written as ordinary differential equations (ODEs), are widely used throughout the sciences, but too often they include overly simplistic assumptions regarding time delays and the duration of time individual entities spend in specific states. In this paper, we illustrate how those assumptions can be refined within the context of the GLCT, and used to derive new models. We have derived multiple new dynamical systems models based on existing models taken from the literature, to illustrate the relative ease of deriving such models using the GLCT. These examples span a range of biological application areas and different model types (DDEs, ODEs, and integro-differential equations), which reflect a mix of different implicit and explicit delay assumptions. Using the GLCT, those delay assumptions were replaced or generalized to yield new ODEs that incorporate phase-type distributed delays and dwell times.

These straightforward generalizations illustrate how modelers can incorporate phase-type distributed delays and dwell times into ODE models, and how those underlying (implicit) stochastic model assumptions are reflected in the corresponding mean field ODE model structure. Importantly, these alternative model formulations can also be used in the computational and mathematical analysis of models that only assume Erlang distributions and could otherwise be derived using the standard linear chain trick (LCT). For an example, see \citet{HurtadoRichards2020}, where we illustrate the potential computational benefits of using a GLCT formulation of models with Erlang dwell time assumptions when computing numerical solutions to such models. 

These generalized models also lay the groundwork for incorporating Coxian, hyperexponential, hypoexponential (i.e, generalized Erlang) and other phase-type distributions into these and similar mean field ODE models. Statistical tools such as \texttt{BuTools} \citep{BuTools2,BuToolswww} allow modelers to fit phase-type distributions to data, thereby allowing modelers to build approximate empirical distributions into ODE models using the GLCT. However, it is important to note that there are certain limitations to approximating some delay or dwell time assumptions with phase-type distributions. For example, delay distributions with compact support (e.g., a continuous uniform distribution) may not be well approximated by phase-type distributions. 

In closing, we hope these new techniques prove to be helpful to modelers in their efforts to build better models, to check the consequences of certain simplifying assumptions, and to gain better intuition for how underlying assumptions are reflected in the structure of ODE model equations.


\subsection*{Acknowledgements}
The authors thank Dr. Deena Schmidt for conversations and comments that improved this manuscript. The authors also thank the organizers and sponsors of the Second International Conference on Applications of Mathematics to Nonlinear Sciences (AMNS-2019), held June 27-30, 2019, in Pokhara, Nepal, and the editors of this Thematic Issue.

\subsection*{Funding}
This work was supported by a grant awarded to PJH by the Sloan Scholars Mentoring Network of the Social Science Research Council with funds provided by the Alfred P. Sloan Foundation; and this material is based upon work supported by the National Science Foundation under Grant No. DEB-1929522.

\subsection*{Disclosure statement}
The authors declare that they have no conflict of interest.

%
%
%

\printbibliography

\end{document}